\documentclass[a4paper, 11pt]{article}

\usepackage{amssymb, amsmath, amsthm}
\usepackage{fullpage}
\usepackage{graphicx}
\usepackage{url}
\usepackage{comment}
\theoremstyle{plain}

\newtheorem{theorem}{Theorem}[section]
\newtheorem{lemma}[theorem]{Lemma}
\newtheorem{corollary}[theorem]{Corollary}

\newtheorem{example}[theorem]{Example}

\numberwithin{equation}{section}

\theoremstyle{definition}
\newtheorem{definition}[theorem]{Definition}
\newtheorem{dassumption}[theorem]{Dispersion Assumption}
\newtheorem{sdassumption}[theorem]{Strengthened Dispersion Assumption}

\theoremstyle{remark}
\newtheorem{remark}[theorem]{Remark}

\newcommand{\R}{{\mathbb R}}

\newcommand{\Prob}{\mathop{\mathbb{P}}\nolimits}
\newcommand{\E}{\mathop{\mathbb{E}}\nolimits}

\newcommand{\sQ}{{\mathcal Q}}
\newcommand{\sP}{{\mathcal P}}
\newcommand{\sF}{{\mathcal F}}

\newcommand{\sS}{{\mathcal S}}
\newcommand{\sL}{{\mathcal L}}

\newcommand{\sD}{{\mathcal D}}
\newcommand{\sE}{{\mathcal E}}

\newcommand{\sM}{{\mathcal M}}

\newcommand{\tp}{{p_R^{-1}}}
\newcommand{\tq}{{q_L^{-1}}}

\title{Robust price bounds for the forward starting straddle}
\author{David Hobson\thanks{Department of Statistics, University of Warwick, 
Coventry, CV4 7AL, UK, d.hobson@warwick.ac.uk} 
\and 
Martin Klimmek\thanks{Nomura Centre for Mathematical Finance, Mathematical Institute, University of Oxford, \newline Oxford, OX1 3LB, UK, martin.klimmek@maths.ox.ac.uk}}

\date{\today}

\begin{document}

\maketitle

\begin{abstract}
In this article we consider the problem of giving a robust,
model-independent, lower bound on the price of a forward starting
straddle with payoff $|F_{T_1} - F_{T_0}|$ where $0<T_0<T_1$. Rather than
assuming a model for the underlying forward price $(F_t)_{t \geq 0}$, we
assume that call prices for maturities $T_0<T_1$ are given and
hence that the marginal laws of the underlying are known. The primal
problem is to find the model which is consistent with the observed call
prices, and for which the price of the forward starting straddle is
minimised. The dual problem is to find the cheapest semi-static subhedge.

Under an assumption on the supports of the marginal laws, but no 
assumption that the laws are atom-free or in any other way regular, we 
derive explicit expressions for the coupling which minimises the price of 
the option, and the form of the semi-static subhedge.

\end{abstract}

\section{Introduction} 

We consider the problem of constructing martingales with two given marginals 
which minimise the expected value of the modulus of the increment. This 
problem has a direct correspondence to the problem 
in mathematical finance of giving a no-arbitrage lower bound on the price of 
an at-the-money forward starting straddle, given today's vanilla call prices 
at the two relevant maturities. There is also a related dual problem, which is 
to construct the most expensive semi-static hedging strategy which 
sub-replicates the payoff of the forward starting straddle for any realised 
path of the underlying forward price process. Under a certain assumption on 
the distributions we solve the primal and dual 
problem and demonstrate that there is no duality gap.

The results of this article complement previous results by Hobson and 
Neuberger \cite{HobsonNeuberger:2008}. In that article, the authors solved the 
analogous problem of constructing no-arbitrage upper price bounds and 
semi-static 
super-replicating hedging strategies for the forward starting straddle. Hobson 
and Neuberger \cite{HobsonNeuberger:2008} give some examples, but the main 
contribution is an existence result, and a proof that the primal and dual 
problems yield equal values. 

Returning to the lower bound case, it follows from results of Beiglb{\"o}ck et 
al~\cite{BeiglboeckHLP:2011} that there exists a solution to the primal 
problem 
and that there is no duality gap. However, Beiglb{\"o}ck et 
al~\cite{BeiglboeckHLP:2011} give an example to show that the dual supremum may 
not be attained. 
Our contribution is to show
that, under a critical but natural condition on the starting and terminal law,
there is an explicit construction of all the quantities of interest. (A
parallel construction gives an explicit form for the optimal martingale in
the upper bound setting, but in that case the condition on the measures is
less natural.) 

There has been a recent resurgence of interest in problems of this type, in 
part motivated by the connection with finance, see the survey article by 
Hobson~\cite{Hobson:2011}, and in part motivated by the connections with the 
optimal
transport problem, see \cite{BeiglboeckHLP:2011, BeiglboeckJuillet:2012,
HLTouzi:2013}. A first strand of the mathematical finance literature is 
concerned with 
robust pricing of particular exotic derivatives, for example lookback 
options~\cite{Hobson:1998, GalichonHLPTouzi:2013}, barrier 
options~\cite{BrownHobsonRogers:2001a, BrownHobsonRogers:2001b, CoxWang:2013},
Asian options~\cite{DDGKV:2002}, basket 
options~\cite{HobsonLaurenceWang:2005,d'AspremontElGhaoui:2006} and volatility 
swaps~\cite{CarrLee:2010,Kahale:2009, HobsonKlimmek:2011}. This strand of the 
literature often makes use of connections with the Skorokhod embedding 
problem, but in some senses the problem here is simpler in that the option 
payoff only depends on the joint law of $(F_{T_0},F_{T_1})$ and is otherwise 
path-independent. For this reason the full machinery of the Skorokhod 
embedding 
problem is not needed, although it can still help with the intuition. A second 
strand of the mathematical finance literature on robust 
pricing~\cite{Cousot:2007, 
DavisHobson:2007,ABPS:2013,DolinskySoner:2013} considers consistency between 
options and 
no-arbitrage conditions.

The optimal transport literature (see Villani~\cite{Villani:2009}) is 
concerned with the cheapest way to transport `sand' distributed according to a 
source measure $\mu$ to a destination measure $\nu$. Such problems are also 
motivated by economic questions and the transport of mass becomes an issue of 
the optimal allocation of economic goods. Unsurprisingly, questions of this 
type were of particular importance to mathematicians working on questions of 
efficiency in planned economies, see for instance 
Kantorovich~\cite{Kantorovich:1975}. The Kantorovich relaxation of Monge's 
original problem goes back to his use of linear programming methods and as 
such, the Lagrangian approach taken in this paper is most natural. With 
respect to the classical optimal transport literature, the novelty of the 
current problem, as elucidated in Beiglb{\"o}ck et 
al~\cite{BeiglboeckHLP:2011} and developed in Beiglb{\"o}ck and Juillet 
\cite{BeiglboeckJuillet:2012} and Henry-Labord\`{e}re and 
Touzi~\cite{HLTouzi:2013} is to add a martingale requirement to the transport 
plan, which is motivated by the idea that no-arbitrage considerations equate 
to a condition that forward prices are martingales under a pricing measure. 
Both \cite{BeiglboeckJuillet:2012} and \cite{HLTouzi:2013} consider the 
problem of minimising the martingale transport cost for a class of cost 
functionals , but the payoff $|y-x|$ is not a member of this class. In 
contrast, here we focus exclusively on the payoff $|y-x|$. This functional 
encapsulates the payoff of a forward starting straddle, which is an important 
and simple financial product which fits into the general framework, and it is 
the original Monge cost function in the classical set-up. For these two 
reasons this cost functional is of significant interest.

\section{Motivation and preliminaries} 

Let $X,Y$ be real-valued random variables 
and suppose $X \sim \mu$ and $Y 
\sim \nu$. We say that the bivariate law $\rho$ is a martingale coupling 
(equivalently a martingale transference plan) 
of $X$ 
and $Y$, and write $\rho \in {\sM}(\mu,\nu)$, if $\rho$ has marginals $\mu$ 
and 
$\nu$ and is such that $\int_y (y - x) \rho( dx, dy) = 0$ for each $x$.
For a univariate measure $\zeta$ define $C_\zeta$ via $C_\zeta(x) = \int
(y-x)^+ \zeta(dy)$. It is well known (see, for instance, 
Strassen\cite{Strassen:67})
that ${\sM}(\mu,\nu)$ is non-empty if and 
only if $\mu$ is less than or equal to $\nu$ in convex order, or 
equivalently $\mu$ and $\nu$ have equal means and $C_{\mu}(x) \leq 
C_{\nu}(x)$ for all $x \in \R$. Then, under the assumption that 
${\sM}(\mu,\nu)$ is non-empty, 
our goal is to find
\begin{equation}
\label{eqn:primal}
{\sP}(\mu,\nu) := \inf_{\rho \in {\sM}(\mu,\nu)} \E[|Y-X|] \equiv  
\inf_{\rho \in 
{\sM}(\mu,\nu)} \int |y-x| \rho(dx,dy).  
\end{equation}

The financial significance of this result is as follows. Let $F=(F_t)_{t \geq 
0}$ denote the forward price process of a financial asset.
Let $T_0$ and 
$T_1$ be two future times with $0<T_0 < T_1$. A well known argument due to 
Breeden and Litzenberger \cite{BreedenLitzenberger:78}, shows that knowledge 
of a continuum of call prices for a fixed expiry is equivalent to knowledge of 
the marginal law at that time. Suppose then that a continuum in strike of call 
prices are available from a financial market and hence that it is possible to 
infer the marginal laws of $F$ at time $T_0$ and time $T_1$. Suppose these 
laws are given by $\mu$ and $\nu$. The problem is to minimise 
$\E[|F_{T_1}-F_{T_0}|]$ over all martingale models for $F$ with the given 
marginals, i.e. to find (\ref{eqn:primal}). We will call this problem 
the {\it primal problem}.

Conversely, suppose we can 
construct a trio of functions $(\psi_0, \psi_1,\delta)$ such that
\begin{equation} \label{eq:hedge}
|y-x| \geq \psi_1(y)-\psi_0(x) + \delta(x)(x-y) \hspace{10mm} \forall 
x,y \in 
\R.
\end{equation}
Then $\E[|Y-X|] \geq \E[\psi_1(Y)] - 
\E[\psi_0(X)] 
=
\int \psi_1(y)\nu(dy) - \int \psi_0(x) \mu(dx)$. It follows from 
(\ref{eq:hedge}) that if $\tilde{\sS}$ is the 
set of trios of functions $(\psi_0, \psi_1,\delta)$ such that 
(\ref{eq:hedge}) 
holds 
then
\begin{equation}
\label{eqn:dual}
 \E[|Y-X|] \geq \sup_{ (\psi_0, \psi_1 ,\delta)  \in \tilde{\sS} }   
\left\{ \int \psi_1(y)\nu(dy) - \int \psi_0(x) \mu(dx) \right\} =: 
{\tilde{\sD}}(\mu,\nu).
\end{equation}
We will call the problem of finding the right-hand-side of (\ref{eqn:dual}) 
the {\it dual problem}. 

Again there is a direct financial interpretation of (\ref{eqn:dual}).
Under our assumption that a continuum of calls is traded, the European 
contingent claims $\psi_1(F_{T_1})$ and $\psi_0(F_{T_0})$ can be 
replicated with portfolios of call options bought and sold at time zero.
Moreover, if the agent sells forward $\delta(F_{T_0})$ units over
$[T_0,T_1]$ then the gains are $\delta(F_{T_0}) (F_{T_0} - F_{T_1})$. 
Combining the two-elements of the semi-static strategy 
(\cite[Section 2.6]{Hobson:2011}) consisting of calls and a simple 
forward position 
yields
\[ \psi_1(F_{T_1}) - \psi_0(F_{T_0}) + \delta(F_{T_0}) (F_{T_0} - 
F_{T_1}) \]
which corresponds to the right-hand-side of (\ref{eq:hedge}).
If (\ref{eq:hedge}) holds then the semi-static strategy is a subhedge 
for the payoff $|F_{T_1}- F_{T_0}|$.

It is clear that if (\ref{eq:hedge}) holds then we must have $\psi_1(x) 
\leq \psi_0(x)$ and that if we want to find $\psi_i$ to maximise the 
right-hand-side of (\ref{eqn:dual}) then we want $\psi_0$ as small as 
possible. As a result a natural candidate for optimality is 
to take $\psi_0=\psi_1$ and the 
problem of finding a trio $(\psi_0, \psi_1, \delta)$ reduces to finding 
a pair $(\psi,\delta)$. We let $\sS$ be the set of pairs $(\psi,\delta)$ 
such that $|y-x| \geq \psi(y)-\psi(x) + \delta(x)(y-x)$ for all $x,y$. 
Define
\begin{equation}
\label{eq:defdual}
\sD(\mu,\nu) = \sup_{\psi,\delta \in \sS} \left\{ \int \psi(y)\nu(dy) - \int 
\psi(x) \mu(dx) \right\} .
\end{equation}

Weak duality gives that $\sP(\mu,\nu) \geq {\tilde{\sD}}(\mu,\nu) \geq 
\sD(\mu,\nu)$. 
Note that there can be no uniqueness of the dual optimiser: if $(\psi,\delta) 
\in \sS$ then so 
is $(\psi + ax +b, \delta - a)$. Hence we may choose any convenient 
normalisation such as $\psi(x_0)=0=\delta(x_0)$ for some $x_0 \in \R$.

In this article, in addition to requiring that $\mu$ and $\nu$ are increasing 
in convex order we will make the following 
additional assumption. 

\begin{dassumption} \label{ass:disjoint} The marginal distributions
$\mu$ and $\nu$ are such that
the support of $\eta : =(\mu-\nu)^+$ is contained in an interval $E$ 
and the support of 
$\gamma := (\nu-\mu)^+$ is contained in $E^c$.
\end{dassumption}

Weak duality gives that $\sP(\mu,\nu) \geq {\tilde{\sD}}(\mu,\nu) \geq
\sD(\mu,\nu)$.
Our goal in this article is to show that $\sP(\mu,\nu) = \sD(\mu,\nu)$ 
and to give explict expressions for the optimisers $\rho$, $\psi$ and 
$\delta$.

\begin{remark}
Note that we make no other regularity assumptions on the
measures $\mu$ and
$\nu$.
For example we do not require that $\mu$ and $\nu$ have densities.
In contrast, Henry-Labord\`{e}re and
Touzi~\cite{HLTouzi:2013} assume that $\mu$ has no atoms. This is 
also a 
simplifying assumption in part of Beiglb{\"o}ck and Juillet
\cite{BeiglboeckJuillet:2012}. Conversely, the example in
Beiglb{\"o}ck et al~\cite{BeiglboeckHLP:2011} which 
shows that the dual optimiser may not exist makes essential use of the 
fact 
that $\mu$ has atoms. In our case, under Assumption~\ref{ass:disjoint} we show 
that a dual maximiser exists 
whether or not $\mu$ has atoms.
\end{remark}

Note that if $\mu$ is less than or equal to $\nu$ in convex order and 
Assumption~\ref{ass:disjoint} holds then $E$ must be a finite interval. $E$ 
may be closed, or open, or half-open. The rationale for 
Assumption~\ref{ass:disjoint} will become apparent in the development of the 
results below. Let us, however, briefly point out that the assumption is 
natural in contexts most commonly encountered in mathematical finance. For 
instance, if the two distributions are increasing in convex order and 
log-normal, then Assumption~\ref{ass:disjoint} is trivially satisfied.

\begin{example} Let $Z$ be a random variable which is symmetric about zero and 
which has density $f_Z$ such that $zf_Z(z)$ is unimodal on $\R_+$. For 
$s<t$ let $X \equiv sZ$ and $Y \equiv tZ$. Let $\mu$ and $\nu$ be the laws of 
$X$ and $Y$. Then Assumption~\ref{ass:disjoint} holds. 
\end{example}

\subsection{Heuristics and motivation for the structure of the solution}

Let $B=(B_t)_{t \geq 0}$ be a standard Brownian motion and consider the 
problem of maximising or minimising $\E[|B_\tau|]$ over all stopping 
times $\tau$, subject to the constraint that $\E[B_\tau^2]=\E[\tau]=1$. 
For the maximum, the solution is a two point distribution at $\pm 1$. 
This is consistent with the solution derived in Hobson and Neuberger
\cite{HobsonNeuberger:2008} for the forward starting straddle, where in 
the atom-free case the solution to the problem of maximising $\E[|Y-X|]$ 
subject to $X \sim \mu$, $Y \sim \nu$ and $\E[Y | X]=X$ is characterised 
by a pair of increasing functions $f,g$ with $f(x)<x<g(x)$, such that, 
conditional on the initial value of the martingale being $x_0$, the 
terminal value of the martingale lies in $\{f(x_0), g(x_0) \}$. 
Unfortunately, the condition that $f$ and $g$ are increasing is not 
sufficient to guarantee optimality and a further `global consistency 
condition' (\cite[p42]{HobsonNeuberger:2008}) is required.

A solution to the problem of minimising $\E[|B_\tau|]$ over 
stopping times such that $\E[B_\tau^2]=1$ does not exist. However, 
$\E[|B_\tau|]$ can be made small by placing some mass at $\pm n$ and a 
majority of the mast at $0$. For instance, placing an atom of size 
$\frac{1}{2n^2}$ at $\pm n$ and the remaining mass $1-1/2n^2$ at $0$, we have 
$\E[|B_\tau|]=1/n$. The intuition which carries over into the minimisation 
problem for the forward starting straddle is therefore to move as little mass 
as possible. In other words, we expect that subject to the same constraints as 
above, $\E[|Y-X|]$ is minimised if $Y \in \{p(X),X,q(X)\}$ for 
functions $p,q$ with $p(x) \leq x \leq q(x)$.

\subsection{Intuition for the dual problem: the Lagrangian formulation}
\label{ssec:intuition}

As in \cite{HobsonNeuberger:2008} we take a Lagrangian approach. The problem 
is 
to minimise
\( 
\int \int |y-x| \rho(dx,dy),
\) 
subject to the martingale and marginal conditions
$\int_x \rho(dx,dy)=\nu(dy)$, $\int_y \rho(dx,dy)=\mu(dx)$,
and 
$\int_y(x-y) \rho(dx,dy) = 0$.

Letting $\alpha(y)$, $\beta(x)$ and $\theta(x)$ denote the multipliers for 
these constraints, define the Lagrangian objective function
$L(x,y) = |y-x|-\alpha(y)-\beta(x)-\theta(x)(x-y)$.
Then the Lagrangian formulation of the primal problem is to minimise over all 
measures $\rho$ on $\R^2$, 
\begin{equation} \label{eq:lagrange}
\int \int \rho(dx,dy) L(x,y) +\int \alpha(y)\nu(dy) + \int \beta(x) \nu(dx).
\end{equation}
For a finite optimum we require $L(x,y) \geq 0$ and we expect that $L(x,y)=0$ 
for 
$y \in \{p(x),x,q(x)\}$, for a pair of functions $(p,q)$, to be determined, 
see Figure~\ref{fig:1}. In 
particular, since $L(x,x)=0$, we expect that $\beta=-\alpha$.
In this case we write $L=L^{\alpha,\theta}$ where
\begin{equation} 
\label{eqn:Ldef}
L^{\alpha,\theta}(x,y) = |y-x| - \alpha(y) + \alpha(x) - 
\theta(x)(x-y) 
\end{equation}

\begin{figure}[ht!]\label{fig:1}
\begin{center}
\includegraphics[height=6cm,width=10cm]{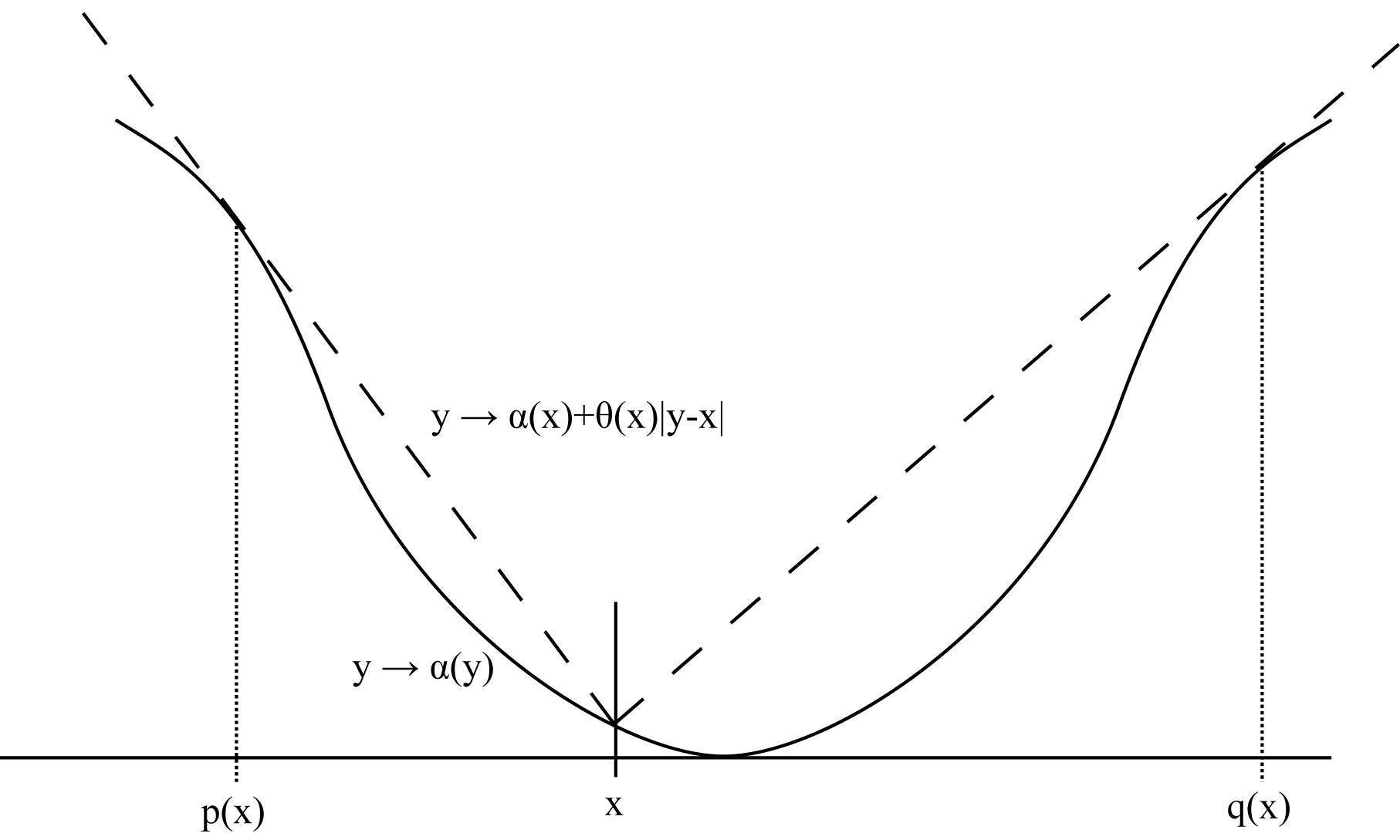}
\end{center}
\caption{We have that $\alpha(x) + \theta(x)|y-x| \geq \alpha(y)$ with
equality for $y \in \{p(x),x,q(x) \}$.}
\end{figure}

Assuming that $\alpha$, $\theta$, $p$ and $q$ are suitably differentiable 
we can
derive expressions for $\alpha$ and $\theta$. 
The conditions $L(x,p(x))=0$, $L_y(x,p(x))=0$, $L(x,q(x))=0$ and 
$L_y(x,q(x))=0$ 
lead directly to the following equations for $x \in E$:
\begin{eqnarray} \label{eq:Lequations}
x-p(x)+\alpha(x)-\alpha(p(x))-\theta(x)(x-p(x))&=&0 \label{eq:Le1} \\
-1-\alpha'(p(x))+\theta(x) &=& 0 \label{eq:Le2} \\
q(x)-x+\alpha(x)-\alpha(q(x))-\theta(x)(x-q(x))&=&0 \label{eq:Le3} \\
1-\alpha'(q(x))+\theta(x) &=& 0 \label{eq:Le4}
\end{eqnarray}
Differentiating (\ref{eq:Le1}) and using (\ref{eq:Le2}), we obtain 
\begin{equation} \label{eq:Le5}
1+\alpha'(x)-\theta(x)-\theta'(x)(x-p(x))=0.
\end{equation}
Similarly, using (\ref{eq:Le3}) and (\ref{eq:Le4}), we obtain 
\begin{equation} \label{eq:Le6}
-1+\alpha'(x)-\theta(x)-\theta'(x)(x-q(x))=0.
\end{equation}
Subtracting the second of these two equations from the first, it follows that 
$\theta'(x)=\frac{2}{q(x)-p(x)}$. Hence $\theta$ is increasing and for 
some $x_0 \in E$, 
\[\theta(x)=\theta(x_0)+\int_{x_0}^x \frac{2}{q(z)-p(z)}dz.\]
Now by adding (\ref{eq:Le5}) and (\ref{eq:Le6}) we have  
$\alpha'(x)=\theta(x)+\frac{\theta'(x)}{2}(2x-p(x)-q(x))$. Then 
\[\alpha(x)=x \theta(x) -\int_{x_0}^x 
\frac{q(z)+p(z)}{q(z)-p(z)}dz + \alpha(x_0)-x_0 \theta(x_0).\] 
Moreover, from the fact that $q(x)$ is a local minimum in $y$ of $L(x,y)$ we 
expect $0 \leq L_{yy}(x,q(x)) = - \alpha''(q(x))$. Hence $\alpha$ is concave at 
points $y = q(x)$.

Fix $x' < x''$.
We must have that for $y \geq x'$, $\alpha(y) \leq \alpha(x') + (y-x')(1+ 
\theta(x'))$, with equality at $y=q(x')$. We also know that for $y \geq x''$,
$\alpha(y) \leq \alpha(x'') + (y-x'')(1+ 
\theta(x''))$, with equality at $y=q(x'')$.
Suppose that $x'' \leq q(x')$. Then
\[ \alpha(x') + (q(x')-x')(1+
\theta(x')) = \alpha(q(x')) \leq \alpha(x'') + (q(x') -x'')(1+  
\theta(x'')). \] 
Similarly, since $x' < x'' \leq q(x'')$, we have
\[ \alpha(x'') + (q(x'')-x'')(1+
\theta(x'')) = \alpha(q(x'')) \leq \alpha(x') + (q(x'') -x')(1+
\theta(x')). \]
Then adding and simplifying we find
$(q(x') - q(x''))(\theta(x'')-\theta(x')) \geq 0$. But $\theta$ 
is increasing and thus 
$q(x') \geq q(x'')$.
Hence $q$ is 
locally decreasing, in the sense that if $x'<x''$ then either 
$q(x') < x''$ or $q(x') \geq q(x'')$.

Now suppose we add Assumption~\ref{ass:disjoint}. The 
pair $(x,q(x))$ must be such that $x \in E, q(x) \in E^c$. Then, if $x' 
\leq x'' \in E$, we cannot have $q(x') \leq x''$. Thus $q(x') \geq 
q(x'')$ and $q$ is necessarily globally decreasing. Similarly $p$ is 
globally decreasing.

\section{Determining $(p,q)$}

\subsection{A differential equation for $p$ and $q$.}
\label{ssec:de}

Suppose that Assumption~\ref{ass:disjoint} is in force.
Suppose that $E$ has end-points 
$\{a,b\}$, either of which may or may not belong to $E$. Then $p:E \mapsto 
(-\infty,a]$ and $q:E \mapsto 
[b,-\infty)$ are decreasing functions. 

Recall the definitions
$\eta : =(\mu-\nu)^+$ and 
$\gamma := (\nu-\mu)^+$. Our martingale transport of $\mu$ to $\nu$ 
involves leaving common mass $(\mu \wedge \nu)$ unchanged, and otherwise 
mapping $\eta$ to $\gamma$.

Suppose that $\eta$ and $\gamma$ have densities $f_\eta$ and $f_\gamma$. 
Then we expect that $p$ and $q$ are strictly decreasing. Then $(X \neq 
Y, X \leq 
z) = (p(z) < Y \leq a) \cup (q(z) < Y)$. It follows that for $z \in E$ we have 
both
\begin{equation} \label{eq:marginalconstraint}
\int_a^z f_\eta(u) du =\int_{q(z)}^\infty f_\gamma(u) du +\int_{p(z)}^a 
f_\gamma(u) du,
\end{equation}
and, from the martingale property,
\begin{equation} \label{eq:martingaleconstraint}
\int_a^z u f_\eta(u) du = \int_{q(z)}^\infty u 
f_\gamma(u) du +\int_{p(z)}^a u f_\gamma(u) du.
\end{equation}

Then by differentiating equations (\ref{eq:marginalconstraint}) and
(\ref{eq:martingaleconstraint})
we can derive a coupled pair of differential equations for the pair 
$(p,q)$:
\begin{eqnarray}
p'(x) &=& - \frac{q(x)-x}{q(x)-p(x)}
\frac{f_\eta(x)}{f_\gamma(p(x))} \equiv
 \frac{q(x)-x}{q(x)-p(x)} 
\frac{f_\mu(x)-f_\nu(x)}{f_\mu(p(x))-f_\nu(p(x))}, \label{eq:downdiff2} \\
q'(x) &=& 
- \frac{x-p(x)}{q(x)-p(x)}
\frac{f_\eta(x)}{f_\gamma(q(x))} \equiv 
\frac{x-p(x)}{q(x)-p(x)} 
\frac{f_\mu(x)-f_\nu(x)}{f_\mu(q(x))-f_\nu(q(x))} \label{eq:updiff2}. 
\end{eqnarray}
where $f_\mu$ and $f_\nu$ are the densities of $\mu$ and $\nu$ which we 
also assume to exist. Moreover we expect the boundary conditions
$p(a)= \sup \{ x : F_\eta(x) < F_\eta(a) \}$ and $q(a)= \sup \{ x : F_\eta(x) 
< 1 \}$ (and similarly $p(b)= \inf \{ x : F_\eta(x) >0 \}$ and $q(b)= 
\inf \{ x : F_\eta(x)  
> F_\eta(a) \}$).

\begin{example} \label{ex:uniformdiff}
Suppose $\mu \sim U[-1,1]$ and $\nu \sim U[-2,2]$. We obtain 
$p'(x)=\frac{x-q(x)}{q(x)-p(x)}$ and $q'(x)=\frac{p(x)-x}{q(x)-p(x)}$. 
Let 
$h(x)=q(x)-p(x)$ and $j(x)=q(x)+p(x)$. 
Note that 
$j'(x)=q'(x)+q'(-x)=\frac{p(x)-x}{h(x)}+\frac{x-q(x)}{h(x)}=-1$. By 
symmetry we expect that $p(x)=-q(-x)$ and then 
$j(0)=0$ so that $j(x)=-x$. To derive $h$, note that
$q'(x) h(x) =p(x)-x$ and so
$(h'(x)+j'(x))h(x) =j(x)-h(x) - 2x$.
Then $h'(x)h(x) = j(x)-2x = - 3x$
so that 
$h(x)^2=A-3x^2$ for some constant $A>0$. It follows that 
$q(x)=(h(x)+j(x))/2 =\sqrt{q(0)^2- \frac{3}{4}x^2} - x/2$. 
The boundary condition $q(-1)=2$ then gives $q(0)=\sqrt{3}$ and thus
$q(x)=\frac{\sqrt{12-3x^2}-x}{2}$. Then also a second boundary $q(1)=1$ 
is satisfied. Further, $p(x)= \frac{-\sqrt{12-3x^2}-x}{2}$.
\end{example}

\begin{example} \label{ex:moduniform}
Consider a slight modification of the above such that $\mu \sim U[-1,1]$ and 
$\nu \sim \frac{5}{8}U[-2,-1] + \frac{3}{8}U[1,4]$. Then, for $x \in (-1,1)$,
\( p'(x) = - \frac{4}{5} \left( \frac{q(x) - x}{q(x)-p(x)} \right) \)
and
\( q'(x) =  - 4 \left( \frac{x-p(x)}{q(x)-p(x)} \right) \),
with boundary conditions $p(-1) = -1$ and $q(-1)=4$. This pair of coupled 
equations looks hard to solve. Nonetheless we will show using the 
methods of subsequent sections that
\begin{equation}
 p(x) = - \frac{1}{6} \left( 5 + 4x + \sqrt{25-8x-8x^2} 
\right);
\hspace{10mm}
q(x) = - \frac{1}{6} \left( 5 + 4x - 5\sqrt{25-8x-8x^2} \right).
\label{ex:moduniformpq}
\end{equation}
\end{example}

Except when it is possible to find special simplifications it appears 
difficult to solve the coupled differential equations for $p$ and $q$.
Indeed when the support of $\nu$ is unbounded, some care may be needed over 
the behaviour of $p$ and $q$ at the boundaries of $E$.
More generally, if $\mu$ and $\nu$ have atoms then we cannot expect $p$ and 
$q$ to be differentiable, or even continuous. 
Instead we will take an alternative approach via the 
potentials of 
$\mu$ and $\nu$, which leads to a unique characterisation of $(p,q)$.

\section{Deriving the subhedge} 

Motivated by the intuition of the 
previous section let $\sQ$ be the set of pairs of decreasing functions
($p:E \mapsto (-\infty,a],q:E \mapsto [b,\infty))$. 
Using the monotonicity of $p$ and $q$ we can 
extend the domains of $p$ and $q$ to $[a,b]$. Note that we do not assume 
that $p$ and $q$ are 
differentiable or even continuous, or that they are strictly monotonic. 
Thus, in general $p^{-1}$ and $q^{-1}$ are set valued. Nonetheless we can 
define
$\tp$ and $\tq$ to be the right-continuous and left-continuous 
inverses of $p$ and $q$ respectively. Then $\tp(y) = \inf \{ x \in E ; 
p(x) \leq y \}$ with the convention that $\tp(y) = b$ if this set is 
empty. Similarly $\tq(y) = \sup \{ x \in E ;  
q(x) \geq y \}$ and $\tq(y) = a$ if $q(a)<y$.

\begin{definition} 
\label{def:thetaalphadeltapsi}
Fix $x_0 \in [a,b]$. For $(p,q) \in \sQ$ define $\theta = 
\theta_{p,q}:[a,b] \rightarrow 
\R$ and $\alpha = \alpha_{p,q}:[a,b] \rightarrow \R$ via
\begin{eqnarray*}
\theta(x) & = & \int_{x_0}^x \frac{2 dz}{q(z)-p(z)}, \\
\alpha(x) & = & \int_{x_0}^x 
\frac{2x-q(z)-p(z)}{q(z)-p(z)} dz = x \theta(x) - \int_{x_0}^x
\frac{q(z)+p(z)}{q(z)-p(z)} dz .
\end{eqnarray*}
Extend these definitions to $\R$ by defining $\delta = 
\delta_{p,q}: \R \rightarrow
\R$ and $\psi = \psi_{p,q}:[a,b] \rightarrow \R$ via
\begin{eqnarray}
\delta(x) & = & \left\{\begin{array}{ll}
\delta(\tp(x)) \hspace{5mm} \; & \; x< a,  \\
\theta(x) & \;  x \in [a,b],  \\
\delta(\tp(x))  &\; x>b.
\end{array}\right.  \label{eq:deltadef} \\
\psi(x) & = & \left\{\begin{array}{ll}
\alpha(\tp(x))+(\tp(x)-x)(1-\theta(\tp(x)))  &\; x <a, 
\\
\alpha(x) &\; x \in [a,b],  \\
\alpha(\tq(x))+(\tq(x)-x)(-1-\theta(\tq(x)))  &\; x > b.
\end{array}\right.
\label{eq:psidef2}
\end{eqnarray}
\end{definition}

\begin{lemma}
For $x>b$ (respectively $x<a$) the value of $\psi(x)$ does not depend on the 
choice $\tq \in q^{-1}(x)$ (respectively $\tp \in p^{-1}(x)$).
\end{lemma}

\begin{proof}
Let $(z_i)_{i = 1,2}$ satisfy $q(z_i)=x$. Then 
\[ 
\alpha(z_2)+(z_2-x)(-1-\theta(z_2)) - 
\{ \alpha(z_1)+(z_1-x)(-1-\theta(z_1)) \}  = \int_{z_1}^{z_2}  
\frac{2(x-q(z))}{q(z)-p(z)} dz. \]
But this expression is zero since $q(z)=x$ on $[z_1,z_2]$.
\end{proof}

\begin{example} \label{ex:unifhedge} 
We calculate $\theta$ and $\alpha$ in the uniform case 
for which we have found $p$ and $q$ in Example \ref{ex:uniformdiff}. 
Setting 
$x_0=0$, we have
\[
\theta(x)  =  \int_0^x \frac{2}{\sqrt{12-3u^2}} du
=  \frac{1}{\sqrt{3}} \int_0^x \frac{du}{\sqrt{1-u^2/4}}
= \frac{2}{\sqrt{3}} \sin^{-1}\left(\frac{x}{2}\right) \]
and
\[
\alpha(x) = x \theta(x) + \int_0^x \frac{u du}{\sqrt{12-3u^2}}   
=  \frac{2x}{\sqrt{3}} \sin^{-1}\left(\frac{x}{2}\right)
+\frac{2-\sqrt{4-x^2}}{\sqrt{3}}.
\]
\end{example}

Recall the definition in (\ref{eqn:Ldef}) of the Lagrangian, now 
extended 
to functions $\psi$ 
and 
$\delta$ defined on $\R$,
\begin{equation}
\label{eqn:Ldefpsidelta}
L^{\psi,\delta}(x,y)=|y-x|+\psi(x)-\psi(y)+\delta(x)(y-x).
\end{equation}
In Theorem~\ref{thm:Lgeq0} below we show that the Lagrangian based on 
(\ref{eqn:Ldefpsidelta}) is non-negative in general and zero 
for $x \in \{ p(x),x,q(x) \}$. It follows that 
the pair
$(\psi,\delta)$ is a semi-static 
subhedge for the forward starting straddle, and hence that
for any pair of random 
variables with $\E[Y|X]=X$ and $X \sim \mu$, $Y \sim \nu$,
\begin{equation}
\label{eq:lbound}
\E[|Y-X|] \geq \int \psi(y) \nu(dy) - \int \psi(x) \mu(dx).
\end{equation}

Further, given a pair of strictly monotonic functions $p:E \mapsto 
(-\infty,a]$ and $q:E \mapsto [b, \infty)$, a probability 
measure $\mu$ 
with 
density 
$f_\mu$ and a sub-probability measure $\eta \leq \mu$ with support 
contained in $E$ and density
$f_\eta$, we can construct a 
pair $(X, \tilde{Y}=Y_{p,q,\eta,\mu})$ via
$\Prob(X \in dx, \tilde{Y} \in dy) = \rho^*(x,y) dx dy$
where 
\begin{equation*} \rho^*(x,y) = \left\{ \begin{array}{ll}
f_\eta(x) \frac{q(x)-x}{q(x)-p(x)} I_{ \{ y= p(x) \} } & y < x, \\
f_\mu(x) - f_\eta(x)   & y = x, \\
f_\eta(x) \frac{x-p(x)}{q(x)-p(x)} I_{ \{ y= q(x) \} } & y > x.
\end{array} \right. 
\end{equation*}

Certainly $X \sim \mu$.
Let $\tilde{\nu} \equiv {\nu}_{p,q,\eta,\mu} \sim \sL(\tilde{Y})$. Then 
$\rho^*(x,y) dx dy \in \sM(\mu, \tilde{\nu})$ and 
\[ \int \int \rho^*(x,y) |y-x| dy dx = \E[|\tilde{Y}-X|] \geq \sP(\mu, 
\tilde{\nu}) \geq \sD(\mu,  \tilde{\nu}) \geq \int \psi_{p,q}(y) 
\tilde{\nu}(dy) - \int \psi_{p,q}(x) \mu(dx) \] 
and there is equality throughout, since $\rho^*$ is only positive when
$L^{\psi_{p,q},\delta_{p,q}}=0$. In particular,
$\sP(\mu, \tilde{\nu}) = \sD(\mu,\tilde{\nu})$ and there is no duality gap.

Suppose $\mu$ and $\eta$ are fixed, and consider varying $p$ and $q$.
Provided it is possible 
to choose 
${p}$ and ${q}$ 
such that
$\tilde{\nu} \equiv {\nu}_{{p},{q},\eta,\mu} \sim \nu$ then we 
have
\[ \sP(\mu,\nu) \leq \int \int \rho^*(x,y) |y-x| dy dx = \E[|Y-X|]
= \int \psi_{{p},{q}}(y) {\nu}(dy) - \int 
\psi_{{p},{q}}(x) 
\mu(dx) \leq \sD(\mu,\nu). \]
Weak duality gives $\sP(\mu,\nu) \geq \sD(\mu,\nu)$, and hence the 
solutions to the 
primal and dual problems are the same.

\begin{remark}
We have that for any $({p},{q}) \in \sQ$
\[ \inf_{ \rho \in \sM(\mu,\nu) } \E[|Y-X|] \geq
\int \psi_{{p},{q}}(y) {\nu}(dy) - \int
\psi_{{p},{q}}(x) \mu(dx) \]

In financial terms this means that any $(p,q) \in \sQ$ can be used to 
generate a semi-static 
subhedge $(\psi_{p,q},\delta_{p,q})$ and hence a lower bound on the price of 
the derivative. This bound is tight if, given $T_0$-call prices are 
consistent with $F_{T_0} \sim \mu$, and $T_1$-call prices are
consistent with $F_{T_1} \sim \nu$, then $\nu \sim {\nu}_{p,q,\eta,\mu}$.
\end{remark} 

In the case where Assumption~\ref{ass:disjoint} holds we 
show in the next section how to choose 
functions $p$ and $q$ such that ${\nu}_{p,q,\eta,\mu} \sim \nu$ using a 
geometric argument. In fact, in cases where 
$\mu$ has atoms, generically there is no pair of functions $(p,q)$ such that
$\nu_{p,q,\eta,\mu} \sim \nu$ and we have to work with multi-valued functions, 
or 
rather to re-parameterise the problem and use some independent randomisation.  
For the rest of this section we 
concentrate on showing that the choice of multipliers in 
Definition~\ref{def:thetaalphadeltapsi} leads to a non-negative Lagrangian.

\begin{theorem}
\label{thm:Lgeq0}
For all $x, y \in \R$, 
\[L^{\psi,\delta}=|y-x|+\psi(x)-\psi(y)-\delta(x)(x-y) \geq 0,\]
with equality for $x \in E$ and $y \in \{ p(x),x,q(x) \}$.
\end{theorem}

\begin{proof}
Suppose first that $x \in [a,b]$. If $x \leq y \leq b$, then 
\begin{eqnarray*}
L^{\psi,\delta}(x,y) &=& (y-x)+x \theta(x)-\int_{x_0}^x 
\frac{p(z)+q(z)}{q(z)-p(z)} dz 
- y \theta(y) + \int_{x_0}^y \frac{p(z)+q(z)}{q(z)-p(z)} dz + \theta(x)(y-x) 
\\
&=& y-x + y(\theta(x)-\theta(y))+\int_x^y \frac{p(z)+q(z)}{q(z)-p(z)} dz \\
&=& \int_x^y dz\left(1-\frac{2y}{q(z)-p(z)} +\frac{p(z)+q(z)}{q(z)-p(z)}\right) \\
&=& \int_x^y dz \left(\frac{2(q(z)-y)}{q(z)-p(z)}\right) \geq 0,
\end{eqnarray*}
since $q(z) \geq b \geq y$. Now for $a\leq y < x$,
\begin{eqnarray*}
L^{\psi,\delta}(x,y) &=& x-y + y(\theta(x)-\theta(y))+\int_x^y 
\frac{p(z)+q(z)}{q(z)-p(z)} dz \\
&=& \int_y^x dz \left(-1+\frac{2y}{q(z)-p(z)} - \frac{p(z)+q(z)}{q(z)-p(z)} 
\right) \\
&=& \int_y^x dz \frac{2(y-p(z))}{q(z)-p(z)} \geq 0,
\end{eqnarray*}
since $p(z) \leq a \leq y$. 

Next, suppose $y > b$. Let $\gamma=\tq(y)$. Then  
$\psi(y)=(y-\gamma)+\psi(\gamma)+\delta(\gamma)(y-\gamma)$
and $L^{\psi,\delta}(\gamma,y)=0$ by construction.
It follows that
\begin{eqnarray*}
L^{\psi,\delta}(x,y)&=& 
(y-x)+\psi(x)-\psi(\gamma)-(y-\gamma)-\delta(\gamma)(y-\gamma)+\delta(x)(y-x) 
\\
&=& \alpha(x)-x \theta(x)-\alpha(\gamma)+\gamma 
\theta(\gamma)+\gamma-x+y(\theta(x)-\theta(\gamma)) \\
&=& -\int_{x_0}^x \frac{p(z)+q(z)}{q(z)-p(z)} dz + \int_{x_0}^\gamma 
\frac{p(z)+q(z)}{q(z)-p(z)}dz +\int_x^\gamma dz - \int_x^\gamma 
\frac{2y}{q(z)-p(z)} dz \\
&=& \int_x^\gamma \frac{dz}{q(z)-p(z)} 2(q(z)-y).   \label{eqn:Ly>b} 
\end{eqnarray*}
Now either $q(x)<y$ or $q(x)> y$ or $q(x)=y$. In the first case, since $q$ is 
decreasing, $x\geq\tq(y)=\gamma$. Since for $z \in (\gamma,x)$, 
$q(z)\leq y$ we have $L^{\psi,\delta}(x,y) \geq 0$. On the other hand if 
$q(x) > 
y$ then $x \leq \tq(y)=\gamma$ and for $z \in (x,\gamma)$, $q(z) \geq 
y$ and again $L^{\psi,\delta}(x,y) \geq 0$. Finally if $q(x)=y$ then $x \in 
q^{-1}(y)$ and $q(z)=y$ on $(x,\gamma)$ so that $L^{\psi,\delta} = 0$.
Similar arguments apply when $y<a$.

Finally suppose $x \notin [a,b]$. We cover the case $x<a$, the case $x>b$ 
being similar. For $y < \tp(x)$,
\[ 0 \leq L^{\psi,\delta}(\tp(x),y) = \psi(\tp(x)) - \psi(y)
+ \tp(x) - y + \theta(\tp(x))(y - \tp(x)) .\]
But also $L^{\psi,\delta}(\tp(x),x) = 0$, and thus
\[ \psi(x) = \psi(\tp(x)) + \tp(x) - x + \theta(\tp(x))(x 
- \tp(x)). \]
Hence
\begin{eqnarray*}
0 \leq L^{\psi,\delta}(\tp(x),y) & = & x - y + \psi(x) - \psi(y) - 
\theta(\tp(x))(x-y) \\
& \leq & |x - y| + \psi(x) - \psi(y) -
\theta(\tp(x))(x-y) = L^{\psi,\delta}(x,y).
\end{eqnarray*}
For $y \geq \tp(x)$,
\[
L^{\psi,\delta}(x,y) - L^{\psi,\delta}(\tp(x),y) = 2 (\tp(x) - x)>0,
\] 
and hence $L^{\psi,\delta}(x,y) > 0$.
\end{proof}

\begin{corollary} 
\label{cor:Lgeq0}
$L^{\psi,\delta}(x,y)=0$ for all pairs $x \in (a,b)$ 
and $y \in [\lim_{z \downarrow x}q(z), \lim_{z \uparrow x} q(z)]$.
Similarly $L^{\psi,\delta}(x,y)=0$ for $y \in [\lim_{z \downarrow x}p(z), 
\lim_{z \uparrow x} p(z)]$.
\end{corollary}

\begin{proof} 
Suppose that $q$ has a jump at $\hat{x} \in E$ and that $\hat{y} \in
[\lim_{z \downarrow \hat{x}}q(z), \lim_{z \uparrow \hat{x}} q(z)]$.
If $\hat{y}$ is such that $\hat{y} \neq q(\hat{x})$ then we can modify 
$q$ by 
defining $\hat{q}(\hat{x}) = \hat{y}$ and $\hat{q}(x) = q(x)$ otherwise.

Write $\hat{\psi} = \psi_{p,\hat{q}}$ and $\hat{\delta} = 
\delta_{p,\hat{q}}$. Clearly, changing the definition of $q$ at the 
single point $\hat{x}$ makes no difference to the definitions of $\psi$ 
or $\delta$. Then, 
\[ L^{\psi,\delta}(\hat{x},\hat{y}) = 
L^{\hat{\psi},\hat{\delta}}(\hat{x},\hat{y}) =
L^{\hat{\psi},\hat{\delta}}(\hat{x},q(\hat{x})) = 0, \]
the last equality following from Theorem~\ref{thm:Lgeq0}.


\end{proof}

\section{The general case}
The goal in this section is to show how to construct a martingale coupling, 
in the case where Assumption~\ref{ass:disjoint} holds. The aim is to construct appropriate generalised versions of $p$ and $q$, allowing for atoms, in 
such a way that if $X \sim \mu$, $Y \in \{ p(X),X,q(X) \}$ and 
$\E[Y|X]=X$ then $Y \sim \nu$.

Define $D(x) = C_\nu(x)-C_\mu(x)$. 
Then $D$ is a non-negative function such that $\lim_{|x| \rightarrow 
\infty}D(x)=0$ and, because of Assumption~\ref{ass:disjoint}, $D$ is 
locally convex on 
$E^c$ and locally concave on $E$.

\subsection{Case 1: $\mu \wedge \nu=0$}
\label{ssec:casemuwedgenu=0}
 
Suppose first that $\mu \wedge \nu=0$ so that $\eta = \mu$ and $\gamma = \nu$ 
are probability measures.

Note that not both $\eta$ and $\gamma$ can have atoms at $a$. Let $\gamma_a= 
\gamma((-\infty,a]) \in (0,1)$.
Then $\gamma_a = \max \{ D'(a-), D'(a+) \}$.

Let $F_\eta^{-1} : [0,1] \mapsto [a,b]$ be the left continuous inverse 
of $\eta ((-\infty,x]) = \eta([a,x])$. 
For $u \in (0,1)$ define $\sE_u(x) = D(F_{\eta}^{-1}(u)) + (x - 
F_{\eta}^{-1}(u)) (\gamma_a - u) - 
D(x)$ and 
\[ \phi(u) = \sup_{p < F_{\eta}^{-1}(u) < q} \frac{\sE_u(q) - D(p)}{q-p} 
.
\]
In regular cases, $\phi$ is the slope of the line of which is tangent to 
both $D$ on $(\infty,a]$ and $\sE_u$ on $[b,\infty)$.
Let $P(u), Q(u)$ be numbers such that $P(u)\leq a$, $Q(u) \geq b$ and 
\[ \phi(u) =  \frac{\sE_u(Q(u)) - D(P(u))}{Q(u)-P(u)}.
 \]
It is clear from the construction, see also Figure~\ref{fig:2}, 
that $\phi$, $P$ and $Q$ are decreasing functions.

\begin{figure}[ht!]\label{fig:2}
\begin{center}
\includegraphics[height=6cm,width=16cm]{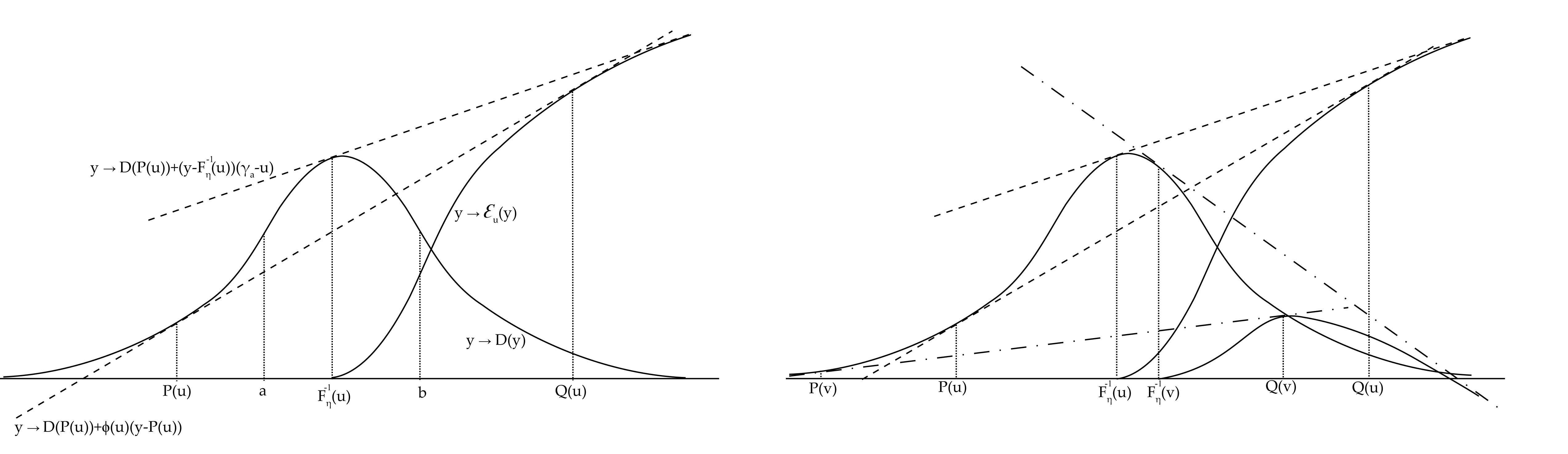}
\end{center}
\caption{The left figure shows, for a single value of $u$, how $\phi(u)$ is 
the slope of the line which is tangent to $D(\cdot)$ at $P(u)$ and 
$\sE_{u}(\cdot)$ at $Q(u)$. The right figure shows how as $u$ increases 
$\sE_u(\cdot)$ decreases and then $P$, $Q$ and $\phi$ also all  
decrease.}
\end{figure}

Define 
\begin{equation}
\label{eqn:zetadef}
\zeta(u)= \inf_{p < F_{\eta}^{-1}(u) < q} \frac{D(q) - \sE_u(p)}{q-p} 
\gamma_a - u - \phi(u). 
\end{equation}
Then $\zeta$ is the slope of the line which is tangent to both ${\sE}_u(x) = 
D(F_{\eta}^{-1}(u)) + (x - F_{\eta}^{-1}(u)) (\gamma_a - u) - D(x)$, defined 
over $x \leq a$ and $D(x)$ defined over $x \geq b$.

Note that if $P(u)$ is unique and $D'$ is continuous at $P(u)$ then $D'(P(u))= 
\phi(u)$. Similarly, if $Q(u)$ is unique and $D'$ is continuous at $Q(u)$ then 
$\sE'_u(Q(u)) = \phi(u)$ which is equivalent to $\zeta(u) = D'(Q(u))$. Since 
$\phi$ and $\zeta$ are 
decreasing 
functions and $D'(P(u))+D'(Q(u)) = \gamma_a - u$, it follows that $\phi$ is 
Lipschitz 
with unit Lipschitz constant. The following lemma shows that this holds 
in general.

\begin{lemma}
$\phi$ is a decreasing Lipschitz function and for $0 \leq u <\hat{u} < 1$, 
$\phi(u) 
\geq \phi(\hat{u}) \geq \phi(u) + (u-\hat{u})$.
\end{lemma}

\begin{proof}
It is clear that $\sE_u(y)$ is decreasing in $u$.
Then, for $\hat{u} > u$,
\[ \phi(u) \geq \frac{\sE_u(Q(\hat{u})) - 
D(P(\hat{u}))}{Q(\hat{u})-P(\hat{u})}
 > \frac{\sE_{\hat{u}}(Q(\hat{u})) - 
D(P(\hat{u}))}{Q(\hat{u})-P(\hat{u})} = 
\phi(\hat{u}), \]
and hence $\phi$ is decreasing. Moreover,
\begin{eqnarray*}
\phi(\hat{u}) \! & \geq & \! \frac{\sE_{\hat{u}}(Q(u)) - D(P(u))}{Q(u) - P(u)} 
\\
              & = & \! \phi(u) + (u-\hat{u}) 
\frac{Q(u)-F_{\eta}^{-1}(u)}{Q(u)-P(u)}
                      + \frac{D(F_{\eta}^{-1}(\hat{u})) - 
D(F_{\eta}^{-1}(u)) + 
(F_{\eta}^{-1}(u)-F_{\eta}^{-1})(\hat{u})(\gamma_a - 
\hat{u})}{Q(u)-P(u)} .
\end{eqnarray*}
By construction the fraction in the middle term lies in $(0,1)$. Further, 
since $D$ is concave on $(a,b)$, $D$ lies below its tangents 
and the numerator in the final term is positive.
Hence $\phi(\hat{u}) \geq \phi(u) + u - \hat{u}$. 
\end{proof}

Since $\phi$ is Lipschitz it is also absolutely continuous and we can write 
$\phi(u) = \gamma_a + \int_0^u \phi'(w) dw$ for some function $\phi'$.

\begin{lemma} We have  
$D(P(u)) - P(u) \phi(u) = \int_u^1 P(v) \phi'(v) dv$, 
$D(Q(u)) - Q(u) \zeta(u) = \int_u^1 Q(v) \zeta'(v) dv$ and 
$D(F_{\eta}^{-1}(u)) - F_{\eta}^{-1}(u)(\gamma_a - u) =
- \int_u^1 F_{\eta}^{-1}(v) dv$.
\end{lemma}

\begin{proof} We prove the first statement, the other two being 
similar.
Given $\{ D(x); x \leq a \}$, for $\theta \in (\gamma_a-1, \gamma_a)$ define the 
conjugate function
$D_c(\theta) = \inf_{x \leq a} \{ D(x) - \theta x \}$. Then  $D_c(0)=0$ and 
$D_c$ is concave so that $D_c(\theta) = \int_0^\theta D_c'(\chi) d \chi$. 
Note that $\phi(u)$ is an element of the subdifferential of $D$ at 
$P(u)$ and hence
\begin{equation*}
D(P(u)) - P(u) \phi(u) = D_c(\phi(u)) = \int_0^{\phi(u)} D_c'(\chi) d \chi 
 = - \int^1_{u} \phi'(v) D_c'(\phi(v)) dv . 
\end{equation*}
Finally note that $D_c'(\phi(v)) = -(D')^{-1}(\phi(v)) = - P(v)$, almost 
everywhere.
\end{proof}

Note that $\phi$ has been constructed to solve
\[ \phi(u)(Q(u)-P(u)) = D(F_{\eta}^{-1}(u)) + 
(Q(u)-F_{\eta}^{-1}(u))(\gamma_a - u) - D(Q(u)) - 
D(P(u)). \]
Using $\phi(u)+ \zeta(u) = \gamma_a - u$ we obtain
\begin{equation}
\label{eqn:characterisationPQ} 
D(P(u)) - P(u) \phi(u) + D(Q(u)) - Q(u) 
\zeta(u) = 
D(F_{\eta}^{-1}(u)) - F_{\eta}^{-1}(u)(\gamma_a - u) 
\end{equation}
and hence
\[ \int_u^1 P(v) \phi'(v) dv + \int_u^1 Q(v) \zeta'(v) dv = - \int_u^1 
F_{\eta}^{-1}(v) dv .
\]
Then we have both $\phi'(u) + \zeta'(u) = -1$ and $P(u) \phi'(u) + 
Q(u)\zeta'(u) 
= -F_{\eta}^{-1}(v)$ Lebesgue almost everywhere on $(0,1)$. Hence
\[ \phi'(u) = - \frac{Q(u)-F_{\eta}^{-1}(u)}{Q(u)-P(u)} \hspace{10mm} 
\mbox{almost everywhere.} \]

\begin{theorem}
\label{thm:coupling}
Let $U$ and $V$ be independent uniform random variables.
Let
$X = F_{\eta}^{-1}(U)$ and $Y =  P(U) I_{A_{(U,V)}} + Q(U)I_{A^C_{(U,V)}}$ 
where
\[ A_{(u,v)} = \left\{ v \leq \frac{Q(u)  
-F_{\eta}^{-1}(u)}{(Q(u)-P(u))} 
\right\}. \]
Then $X$ has law $\eta$ and 
$Y$ has law $\gamma$. 
\end{theorem}

\begin{proof} It is immediate that $X \sim \eta$.
For $y \leq a$ let $P_R^{-1} = \inf \{ u \in (0,1): P(u) \leq y \}$. Then
\begin{eqnarray*} \Prob(Y \leq y) & = & \int_{P_R^{-1}(y)}^1 \frac{Q(u) 
-F_{\eta}^{-1}(u)}{Q(u)-P(u)} du \\
& = &
- \int_{P_R^{-1}(y)}^1 \phi'(u) du = \phi(P_R^{-1}(y)) = D'(y+) = 
\gamma((-\infty,y]). \end{eqnarray*}
A similar argument for $y>b$ gives that $Y \sim \gamma$.
\end{proof}

\subsection{\bf Case 2: $\mu \wedge \nu \neq 0$}  
If $(\mu \wedge \nu)(\R) = 1$ then we must have $Y=X$ and $\E|Y-X|=0$. 
Otherwise, let 
$(\mu \wedge \nu)(\R) = 
\kappa \in (0,1)$. Then $\eta(\R)=\eta(E)= 1- \kappa = \gamma(E^c) = 
\gamma(\R)$. Define probability measures $\overline{\mu \wedge \nu}$, 
$\overline{\eta}$ 
and $\overline{\gamma}$ by 
$\overline{\mu \wedge \nu}(A) = (\mu \wedge \nu)(A)/\kappa$,
$\overline{\eta}(A) = \eta(A)/(1-\kappa)$ and $\overline{\gamma}(A) = 
\gamma(A)/(1-\kappa)$. 

The following result follows immediately from Theorem~\ref{thm:coupling}.


\begin{theorem}
\label{thm:couplingGen}
Let $U$ and $V$ be independent uniform random variables.
and let $Z$ be a Bernoulli random variable with parameter $\kappa$. 
Conditional on
$Z=1$, then set $Y=X \sim \overline{\mu \wedge \nu}$. Otherwise, 
conditional
$Z=0$,
let
$X = F_{\overline{\eta}}^{-1}(U)$ and $Y =  P(U) I_{A_{(U,V)}} + 
Q(U)I_{A^C_{(U,V)}}$
where
\[ A_{(u,v)} = \left\{ v \leq \frac{Q(u)
-F_{\overline{eta}}^{-1}(u)}{(Q(u)-P(u))}
\right\}, \]
where $P$ and $Q$ are defined as in Section~\ref{ssec:casemuwedgenu=0} 
for the disjoint measures $\overline{\eta}$ and $\overline{\gamma}$.
 
Then $X$ has law $\mu$ and
$Y$ has law $\nu$.
\end{theorem}

\subsection{Strong duality}
\begin{theorem}
\label{thm:strongduality}
Suppose $\mu \leq \nu$ in convex order and Assumption~\ref{ass:disjoint} 
holds.
Then $\sP(\mu,\nu) = \tilde{\sD}(\mu,\nu) = \sD(\mu,\nu)$. Moreover, the 
optimal coupling is given by the construction in 
Theorem~\ref{thm:couplingGen}
and the optimal semi-static hedging strategy is given by the expressions in
(\ref{eq:deltadef}) and (\ref{eq:psidef2}), for $(p(x),q(x)) = 
(P(F_{\overline{\eta}}(x)),Q(F_{\overline{\eta}}(x)))$. 
\end{theorem}

\begin{proof}
By Theorems~\ref{thm:coupling} and \ref{thm:couplingGen},  
we have a candidate coupling. Moreover, by 
Theorem~\ref{thm:Lgeq0} and
Corollary~\ref{cor:Lgeq0}, it follows that $L^{\psi,\delta}=0$ wherever 
the candidate coupling 
places mass, and hence, for this choice of joint law for $(X,Y)$ and 
$(\psi,\delta) = (\psi_{p,q},\delta_{p,q})$ we have
$\E[|Y-X|] = \int \psi(y) \nu(dy) - \int \psi(x) \mu(dx)$.
Then $\sP(\mu,\nu) \leq \E[|Y-X|] = \int \psi(y) \nu(dy) - \int \psi(x) 
\mu(dx) \leq \sD(\mu,\nu)$ and by weak duality there is equality throughout.
\end{proof}

\section{Examples}

\subsection{The case with densities}
If $\mu$ and $\nu$ have densities then it is not necessary to introduce $P$ 
and $Q$ and instead we can use the functions $p(x)=P(F_\eta(x))$ and 
$q(x)=Q(F_\eta(x))$. Then from (\ref{eqn:zetadef}) and 
(\ref{eqn:characterisationPQ}) we find that the pair
$(p,q)$ are characterised by 
\begin{eqnarray}
D'(x)&=&D'(q(x))+D'(p(x)), \label{eq:D1} \\
D(x)-xD'(x)&=&D(q(x))-q(x)D'(q(x))+D(p(x))-p(x)D'(p(x)) \label{eq:D2},
\end{eqnarray}
for $x \in E$.

\begin{example} \label{ex:uniform1}
Recall Example~\ref{ex:uniformdiff} and suppose $\mu \sim U[-1,1]$ and $\nu 
\sim U[-2,2]$.
\begin{equation*} \label{eq:VPo}
D(x)= \left\{\begin{array}{ll}
\frac{1}{8}x^2 + \frac{1}{2}x +\frac{1}{2} &\; -2 \leq x \leq -1,  \\
\frac{1}{4} -\frac{1}{8} x^2  &\; -1 < x < 1, \\
\frac{1}{8}x^2 - \frac{1}{2}x +\frac{1}{2}  &\; 1 \leq x \leq 2.
\end{array}\right.
\end{equation*}
By (\ref{eq:D1}), for $-1<x<1$ we have $-x=q(x)+p(x)$. Then, by (\ref{eq:D2}) 
for $x \in 
[-1,1]$ 
we have
\begin{equation*}
D(x)-xD'(x) = D(q(x))-q(x)D'(q(x))+D(-x-q(x))+(x+q(x))D'(-x-q(x)) .
\end{equation*}
After some simplification we find $q(x)^2+q(x)x+x^2-3=0$, whence 
$q(x)=\frac{-x+\sqrt{12-3x^2}}{2}$, and 
$p(x)=\frac{-x-\sqrt{12-3x^2}}{2}$.
\end{example}

\begin{example} \label{ex:moduniform2}
Recall Example~\ref{ex:moduniform} and suppose that $\mu \sim U[-1,1]$ 
and
$\nu \sim \frac{5}{8}U[-2,-1] + \frac{3}{8}U[1,4]$. Then, 
\[ D(x) = \left\{ \begin{array}{ll} 
           \frac{5}{16}(x+2)^2   & -2 \leq x \leq -1; \\
           \frac{11}{16} + \frac{1}{8}x - \frac{1}{4}x^2 \ \   
                      & -1 < x \leq 1; \\   
           \frac{1}{16}(4-x)^2   & 1 < x \leq 4,
          \end{array} \right.
\] with $D \equiv 0$ outside $[-2,4]$.

Then (\ref{eq:D1}) and (\ref{eq:D2}) yield 
\[ 4x + 5p + q +5 = 0,  \hspace{10mm} 4x^2 + 5p^2 + q^2 = 25. \]
Eliminating $q$ and solving for $p$ gives the first expression in 
(\ref{ex:moduniformpq}), from which the second follows.
\end{example}

\subsection{The general case with atoms} 
The full power of the analysis 
is illustrated in the case where $\mu$ has atoms or $\nu$ has intervals 
with no mass. Then there is no hope of constructing the optimal 
martingales via an approach using differential equations.

\begin{example}
Let $\mu$ be any measure supported on $E=[-1,1]$ and suppose $\nu$ has support 
$E^c$ and density $\nu(dx) = |x|^{-3} dx$ for $|x| > 1$. Then for 
$|y|>1$, 
$D(y) = \frac{1}{2|y|}$, $D'(y) = - \frac{1}{2y|y|}$ and 
$D(y)-yD'(y)=\frac{1}{|y|}$. Further, $\gamma_a=\frac{1}{2}$. 

Since $\mu$ is not yet specified we do not have an expression for $D$ on 
$E$. Nonetheless, for $\theta \in (- 
\frac{1}{2}, \frac{1}{2})$ we can define the conjugate $D_c(\theta) = \sup_{-1 
< x < 1} (D(x) - x \theta)$, so that $D_c(\frac{1}{2} - u) = D(F_\mu^{-1}(u)) 
- F_\mu^{-1}(u)(\frac{1}{2} - u)$. 

Then $P=P(u)$ and $Q=Q(u)$ solve
\[ \phi(u) = D'(P) = \sE'_u(Q) = \frac{\sE_u(Q) - D(P)}{Q-P} \]
and hence
\[ \frac{1}{2P^2} = \left( \frac{1}{2} - u \right) + \frac{1}{2Q^2} =
\frac{D_c(\frac{1}{2} - u) + (\frac{1}{2} - u)Q - D(Q) - D(P)}{Q-P}. \]
Then, from the first equality  
\begin{equation}
\label{eqn:PQdefEG} 
\frac{1}{P^2} - \frac{1}{Q^2} = 1- 2u 
\end{equation}
and 
\[ Q \left( \frac{1}{2} - u  + \frac{1}{2Q^2} \right) - P \frac{1}{2P^2} =
D_c\left( \frac{1}{2} - u \right) + \left( \frac{1}{2} - u \right) Q - 
\frac{1}{2Q} + 
\frac{1}{2P}, 
\]
which simplifies to
\begin{equation}
\label{eqn:QPdefEG}
 \frac{1}{|P|} + \frac{1}{Q} =  D_c \left( \frac{1}{2} - u \right) .
\end{equation}
It is straightforward to solve (\ref{eqn:PQdefEG}) and (\ref{eqn:QPdefEG}) 
for $1/P$ and $1/Q$ and 
we find
\[ P(u) = - \frac{2 D_c(\frac{1}{2} - u)}{D_c(\frac{1}{2} - u)^2 + 2 
(\frac{1}{2} - u)} \hspace{10mm}
Q(u) =  \frac{2 D_c(\frac{1}{2} - u)}{D_c(\frac{1}{2} - u)^2 - 2 
(\frac{1}{2} - u)} . \]
For example, suppose $\mu \sim \frac{1}{2} \delta_{-1/2} + \frac{1}{2} 
U[0,1]$. Then, for $x \in (-1,- \frac{1}{2}]$ $D(x)= 1 + x/2$; for $x 
\in (- 
\frac{1}{2}, 0)$, $D(x)= 3/4$; and for $x \in [0,1]$, $D(x)= (3-x^2)/4$. 
Further, for $0 < u \leq 1/2$,
\[ P(u) = - \frac{(2-u)}{2 - 3u + u^2/4} \hspace{10mm}
Q(u) = \frac{(2-u)}{u + u^2/4}; \]
whereas, for $1/2 < u < 1$,
\[ P(u) = - \frac{2(1-u+u^2)}{2 - 4u + 3u^2-2u^3 +u^4} \hspace{10mm}
Q(u) = \frac{2(1-u+u^2)}{3u^2 -2u^3 + u^4}. \]
\end{example}

\begin{example}
Suppose $\mu \sim U(-1,1)$ and $\nu \sim U \{ -2, -1, 3 \}$.
Then $D$ is zero outside $(-2,3)$ and 
$D(x) = \frac{x+2}{3}$ for $-2<x \leq -1$,
$D(x) = \frac{3}{4}  + \frac{x}{6} - \frac{x^2}{4}$ for $-1<x \leq 1$ and
$D(x) = 1- \frac{x}{3}$ for $1 <x < 3$.

It is clear that $q(x)=3$ and $p(x)=-1$ for $x<x^*$ and $p(x)=-2$ for 
$x>x^*$,
where $x^* \in (-1,1)$ is to be found.

Then, see Figure~\ref{fig:3}, 
$x^*$ is the solution to
\[ \left. \left\{ D(x) + (y-x)D'(x) - D(y) \right\} \right|_{y=3} = 
\left. \left\{ D(-1) + (y+1)D'(-1-) \right\} \right|_{y=3} = \frac{5}{3}, \]
and we find $x^* = 3 - 4 \sqrt{\frac{2}{3}}$.


\begin{figure}[ht!]\label{fig:3}
\begin{center}
\includegraphics[height=6cm,width=10cm]{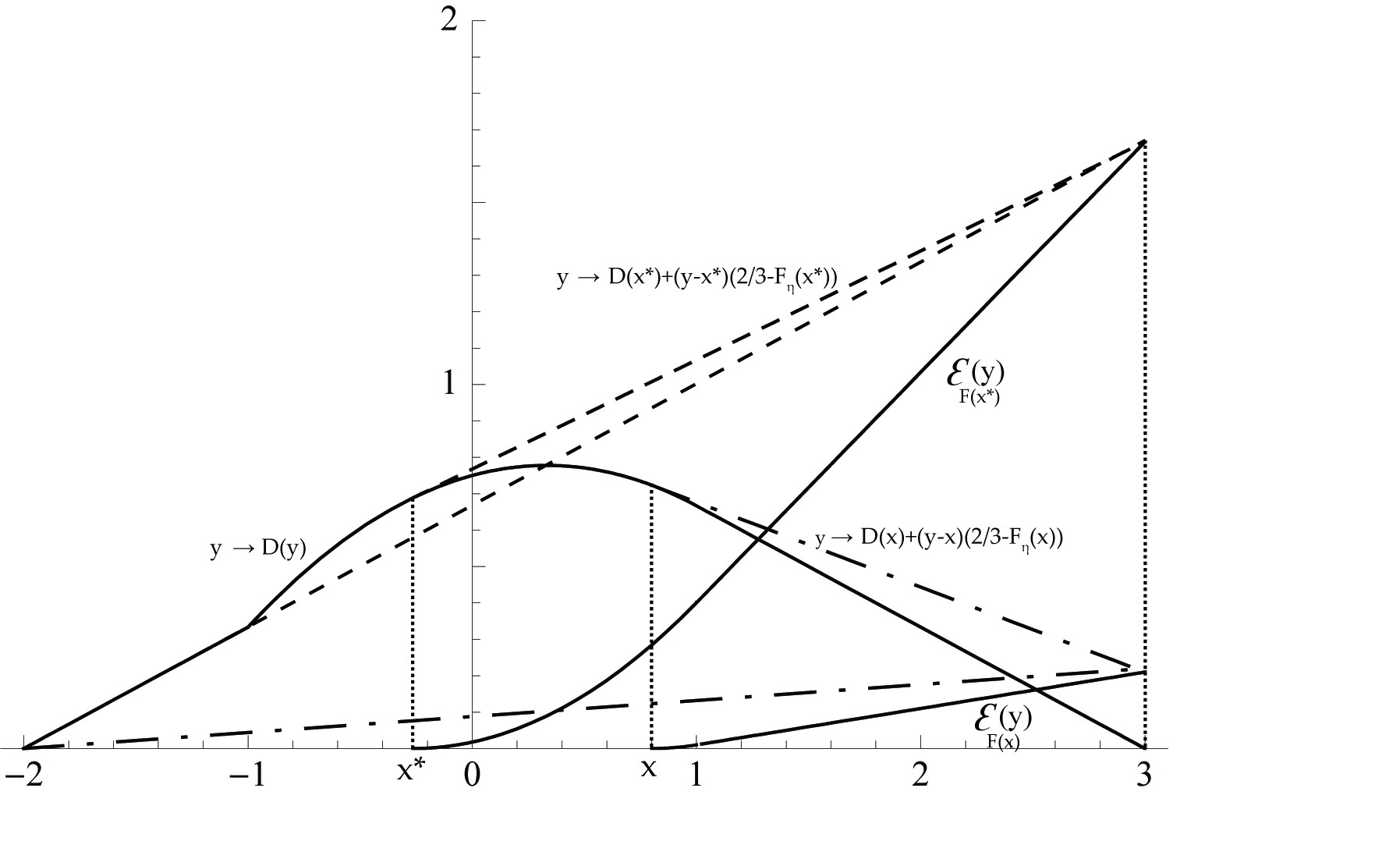}
\end{center}
\caption{The point $x^*$ is such that the tangent to $D(\cdot)$ at $x^*$
meets the line $y = (x+2)/3$ (ie the linear extension of $y
\mapsto D(y)$ taken over $(-2,-1)$) at $x$-coordinate 3. 
For $x>x^*$ the line joining $(-2,D(-2)=0)$ to $(3, 
\sE_{F(x)}(3))$ lies below $D$ on $(-\infty,-1)$ and above $\sE_{F(x)}$ on 
$(1,\infty)$. Hence, for $(x>x^*)$ we have $(p(x)=-2,q(x)=3)$. For $x<x^*$ 
(not shown) 
the line joining $(-1,D(-1)=1/3)$ to $(3,
\sE_{F(x)}(3))$ lies below $D$ on $(-\infty,-1)$ and above $\sE_{F(x)}$ on
$(1,\infty)$. Hence, for $(x>x^*)$ we have $(p(x)=-1,q(x)=3)$.}
\end{figure}

\end{example}

\section{Concluding remarks}

\subsection{Upper bound}
We show below that in the special case where i) $(\mu \wedge \nu)(\R) = 0$, 
ii) $\mu$ has support 
on 
an interval $E$ and iii) $\nu$ has support in $E^c$, a modification of the 
above methods yields the martingale coupling which {\em maximises} 
$\E[|Y-X|]$. Moreover, the optimal transport plan and the super-hedging 
strategies can also be given explicitly as for the lower bound case. 
Recall that Hobson and Neuberger~\cite{HobsonNeuberger:2008} give an 
existence proof, and show that there is no duality gap, but do not 
give a general method for deriving explicit forms for the 
optimal coupling. 

If $\mu$ is a point mass at $x$, then the problem is trivial and $\E[|Y - 
x|] = \int |y-x| \nu(dy)$. Otherwise, proceeding as in 
Section~\ref{ssec:intuition}, but noting that now we expect $L \leq 0$, 
$\beta(x) \geq -\alpha(x)$ and, conditional on $X=x$, mass at two points 
only, we can find equations for monotonic increasing functions (now 
labelled $g$ and $h$, instead of $p$ and $q$). Then $\alpha$ is convex at 
points $y=g(x)\leq x$ and $y=h(x)\geq x$. We find also that $g$ and $h$ 
are increasing. This general structure is evident in 
\cite{HobsonNeuberger:2008}.

Now we add the following assumption:
\begin{sdassumption} \label{ass:disjointU} The marginal distributions
$\mu$ and $\nu$ are such that $(\mu \wedge \nu)(\R) = 0$,
the support of $\mu$ is contained in an interval $E$ and the support of
$\nu$ is contained in $E^c$.
\end{sdassumption}

As before, suppose $E$ has endpoints $\{ a,b \}$. Then necessarily $g:E 
\mapsto (-\infty,a]$ and $h:E \mapsto 
[b,\infty)$. In the case with densities we can write down equations 
similar to 
(\ref{eq:marginalconstraint}) and (\ref{eq:martingaleconstraint}):
\begin{eqnarray} \label{eq:marginalconstraint2}
\int_a^z f_\mu(u) du & = &
\int_{-\infty}^{g(z)} f_\nu(u) du + \int_{b}^{h(z)} f_\nu(u) du \\
\label{eq:martingaleconstraint2}
\int_a^z u f_\mu(u) du & = &
\int_{-\infty}^{g(z)} u f_\nu(u) du + \int_{b}^{h(z)} u f_\nu(u) du. 
\end{eqnarray}

The defining equations for $g$ and $h$ become
\( C_\mu'( x) = C_\nu'( f(x) ) + C_\nu'( g(x) ) - C_\nu'( b ) \)
and
\[ C_\mu(x) - x C_\mu'(x) = C_\nu( f(x) ) - f(x) C'_\nu( f(x) )
   + C_\nu( g(x) ) - g(x) C'_\nu( g(x) ) - [ C_\nu( b ) - b C'_\nu( b )] . 
\] 
This time the equations do not simplify into expressions involving $D(x) 
= C_\nu(x)- C_\mu(x)$ alone.

By analogy with the lower bound case, and now allowing for measures 
which are not absolutely continuous, 
introduce for $u \in (0,1)$ and $y\geq b$,
\[ \sF_u(y) = C_{\mu}(F^{-1}_\mu(u)) + (y- F^{-1}_\mu(u))(u-1)
  + C_\nu(b) + (y-b) C'_\nu(b) - C_\nu(y)    \] 
and define
\[ \varphi(u) = \sup_{g \leq a < b \leq h} \frac{ \sF_u(h) - C_\nu(g)}{h 
- g}. \] 
Then the suprema is attained at $G(u) \leq a < 
b \leq H(u)$ and
\[ \varphi(u) = \frac{ \sF_u(H(u)) - C_\nu(G(u))}{H(u)-G(u)} = C_\nu'(G(u)) = 
\sF_u'(H(u)) \]
assuming these last two derivatives exist.

\begin{theorem}
\label{thm:upperbound} Suppose Assumption~\ref{ass:disjointU} is in force.

Let $U$ and $V$ be independent uniform random variables.
Let
$X = F_{\mu}^{-1}(U)$ and $Y =  G(U) I_{\tilde{A}_{(U,V)}} + 
H(U)I_{\tilde{A}^C_{(U,V)}}$
where
\[ \tilde{A}_{(u,v)} = \left\{ v \leq \frac{H(u)
-F_{\mu}^{-1}(u)}{(H(u)-G(u))}
\right\}. \]
Then $X$ has law $\mu$ and
$Y$ has law $\nu$.

Moreover $\E[|Y-X|] =\sup_{\rho \in \sM(\mu,\nu)} \int \int 
|y-x|\rho(dx, dy)$.
\end{theorem}

\begin{proof}
The proof of the first part of the theorem follows the proof of 
Theorem~\ref{thm:coupling}. The fact that this coupling yields the maximal 
value of
$\int \int |y-x|\rho(dx, dy)$ follows from analogues of 
Theorem~\ref{thm:Lgeq0}, 
Corollary~\ref{cor:Lgeq0} and Theorem~\ref{thm:strongduality}.
\end{proof}

\begin{example}
Suppose $\mu \sim U[-1,1]$ and $\nu\sim \frac{1}{2} U[-3,-1] + 
\frac{1}{2} U[1,3]$. Then $C_{\mu}(x) = (1-x)^2/4$ for $x \in (-1,1)$ 
and $C_{\mu}(x) = x^- = (-x) \vee 0$ otherwise. Further, 
$C_\nu(x) = \frac{(x+3)^2}{8} - x$ for $-3 < x \leq -1$,
$C_\nu(x) = 1-\frac{x}{2}$ for $-1 < x \leq 1$,
$C_\nu(x) = \frac{(3-x)^2}{8}$ for $1 < x \leq 3$ and
$C_{\mu}(x) = x^-$ otherwise.

In this case $g(x) = G(F_\mu(x))$ and $h(x) = H(F_\mu(x))$ are well defined. 
Define
$\tilde{\sF}_x(y) = \sF_{F_\mu(x)}(y)$. Then
\( \tilde{\sF}_x(y) = \frac{1}{8} \left( 1 - 2x^2 - 2y + 4xy - y^2 
\right) \)
and it is easy to check that the pair $(g(x)=x-2,h(x)=x+2)$ solve
\[ C_{\nu}'(g(x)) = \tilde{\sF}'_x (h(x))  = \frac{\tilde{\sF}_x (h(x)) 
- C_\nu(g(x))}{h(x)-g(x)} \]
Hence $\sup_{\rho \in \sM(\mu,\nu)} \int \int |y-x| \rho(dx,dy)$ is 
attained by the coupling $Y = X + \Theta$ where $\Theta$ is uniform on 
$\{ \pm 2 \}$ and is independent of $X$. 

For this example there is a simple proof of optimality. Jensen's inequality 
gives $\E[(Y-X)^2] \geq (\E[|Y-X|])^2$, and the left hand side is independent 
of the martingale coupling and equal to $\int y^2 \nu(dy) - \int x^2 \mu(dx)$.
Hence for our example, $\E[|Y-X|] \leq 2$ for any martingale coupling, and
$\E[|Y-X|]$ is 
easily seen to be equal to 2 for the coupling of the previous paragraph.
\end{example}

\begin{remark}
Unlike in the lower bound case, it is not the case that if $(\mu \wedge 
\nu)(\R)>0$ then the optimal construction is to embed the common part of 
the distribution by setting $Y=X$ and then to embed the remaining part 
by the style of construction given in Theorem~\ref{thm:upperbound}. This 
explains why we need the strengthened dispersion assumption.
\end{remark}
 
\subsection{The situation when the 
Dispersion Assumption does not hold}
Return to the lower bound and
consider the case where Assumption~\ref{ass:disjoint} fails. For simplicity we 
assume that $\mu$ and $\nu$ have densities. Then we expect there to be 
functions $p(x) < x < q(x)$ such that conditional on $X=x$, $Y \in 
\{p(x),x,q(x) \}$. 

From the analysis of Section~\ref{ssec:de} we expect that if $x'<x''$ then 
$p(x'') \notin (p(x'),x')$ and $q(x') \notin (x'',q(x''))$. 
This condition is 
necessary for optimality, but not sufficient and there may be several 
candidate martingale couplings characterised by pairs $(p,q)$ with this 
property. A further condition is necessary, akin to the `global 
consistency 
condition' of \cite{HobsonNeuberger:2008}.

\begin{figure}[ht!]\label{fig:4}
\begin{center}
\includegraphics[height=7cm,width=7cm]{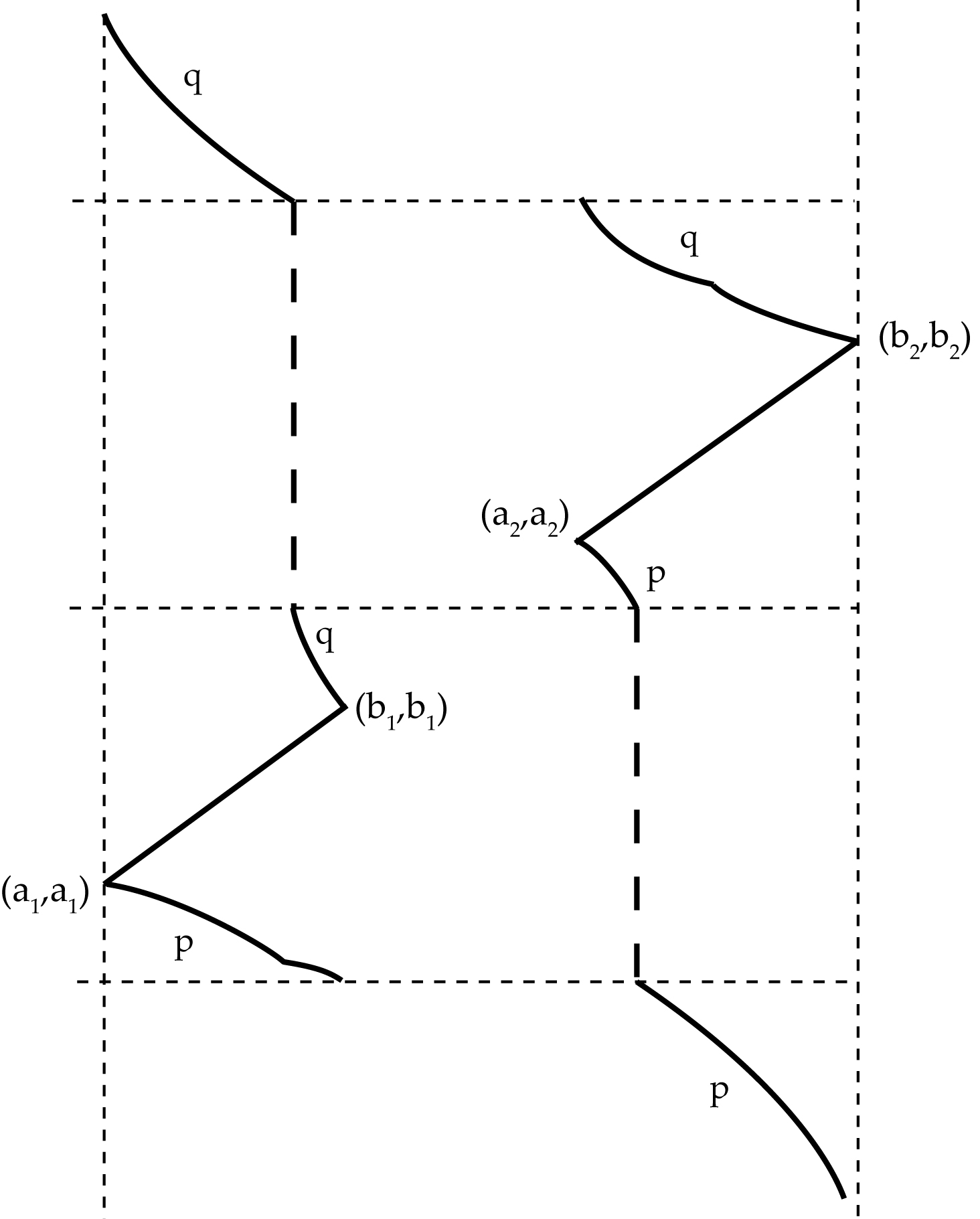}
\end{center}
\caption{A viable pair $(p,q)$, in the relatively simple case where the set 
with $\mu - \nu>0$ consists of two intervals.}
\end{figure}

\begin{example}
Suppose $\mu \sim U[-1,1]$ and $\nu \sim \frac{1}{4} U[-2,-1] + \frac{1}{2} 
\delta_0 + \frac{1}{4} U[1,2]$ where $\delta_0$ denotes a point mass at zero.
Suppose the mass on $[-1,0]$ is mapped to $[c, -1] \cup \{0 \} \cup 
[d, 2]$ where $-2< c<-1$ and $1<d<2$. Then mass and mean considerations imply
$\frac{1}{2} = \frac{(-1-c)}{4} + \theta + \frac{(2-d)}{4}$ and $ 
\frac{-1}{2}\frac{1}{2} = \frac{(-1-c)}{4} \frac{(c-1)}{2} + \frac{(2-d)}{4} 
\frac{(d+2)}{2}$ where $\theta$ is the mass mapped from $[-1,0]$ to  
zero, so $0 \leq \theta \leq 
\frac{1}{2}$. These equations simplify to $c+d = 4 \theta-1$ and $c^2 + d^2 = 
7$. Then
$ d = \frac{(4 \theta - 1)}{2} + \frac{\sqrt{14 - (4 \theta - 1)^2}}{2}$ and 
$c =\frac{(4 \theta - 1)}{2} - \frac{\sqrt{14 - (4 \theta - 1)^2}}{2}$. 
In 
particular 
there is a parametric family of martingale couplings for which $p$ and $q$ 
have the appropriate monotonicity properties.

In this example we expect by symmetry that the optimal solution has $\theta = 
\frac{1}{4}$, and this can indeed be shown to be the case. In general further 
arguments are required.
\end{example}

\subsection{Multi-timestep version}
The analysis of the lower bound extends in a straightforward fashion to 
multiple time-points, provided an analogue of Assumption~\ref{ass:disjoint} 
holds.

Let $T_0=0<T_1< \ldots T_n$ be an increasing sequence of times and suppose 
the marginal distributions $\mu_i$ of a martingale $M$ are known at each 
time $T_i$. 
Consider the problem of giving a lower bound on $S = \E[ \sum_{j=1}^n |M_{T_j} 
- M_{T_{j-1}}|]$.

For the problem to be well-posed we must have that the family $(\mu_i)$ is 
increasing in 
convex order. Let 
$\eta_i = (\mu_{i-1} - \mu_i)^+$ and $\gamma_i = (\mu_i - \mu_{i-1})^+$. 
Suppose further that $\eta_i$ has support contained in an interval $E_i$ and 
$\gamma_i$ has support contained in $E_i^c$. Then the methods of the paper 
apply, and we get a tight lower bound on $S$ of the form $S \geq \sum_{i=1}^n 
\sP(\mu_{i-1},\mu_i)$. 

\subsection{Model independent bounds and mass transport}
There have recently been several papers,
including Beiglb{\"o}ck et
al~\cite{BeiglboeckHLP:2011}, Beiglb{\"o}ck and Juillet
\cite{BeiglboeckJuillet:2012} and Henry-Labord\`{e}re and
Touzi~\cite{HLTouzi:2013} which have made connections between the problem 
of finding optimal model-independent hedges and the Monge-Kantorovich mass 
transport problem, and between the support of the extremal model and 
Brenier's principle.

Beiglb{\"o}ck and Juillet
\cite{BeiglboeckJuillet:2012} introduce a left-curtain martingale 
transport plan. This coupling has many desirable features, and using the 
methods of Henry-Labord\`{e}re and Touzi~\cite{HLTouzi:2013} it is 
possible to calculate the form of the coupling in several examples.
The 
motivations of Beiglb{\"o}ck and Juillet~\cite{BeiglboeckJuillet:2012} in 
introducing 
the left-curtain coupling are at least threefold:
firstly it is relatively tractable; secondly they use it 
to develop an analogue of Brenier's 
principle; and thirdly this coupling 
minimises $\E[c(X,Y)]$ for any cost functional 
$c(x,y)=h(y-x)$ where the third derivative of $h$ is positive. (Note that this 
manifestly excludes the case $h(z)=|z|$ which is the topic of our study.) This 
optimality 
result was extended by Henry-Labord\`{e}re and Touzi~\cite{HLTouzi:2013}
to include all 
functions $c$ which satisfy the Spence-Mirrlees type 
condition 
$c_{xyy}>0$. (One 
example which is natural in finance is the payoff $c(x,y) = - (\log y - \log 
x)^2$ which arises in variance swap payoffs, see also 
\cite{HobsonKlimmek:2011}.)

In addition to the clear financial motivation in terms of forward starting 
straddles, the payoff $|y-x|$ is a natural 
object of study since it is the original Monge cost function in the classical 
(non-martingale) setting, see R\"{u}schendorf~\cite{Ruschendorf:2007}. In that 
setting the non-differentiability makes this payoff a relatively difficult 
functional to study.

The approach in \cite{BeiglboeckJuillet:2012} is based on a notion of cyclic 
monotonicity in a martingale setting. This has the advantage of providing 
existence results in a general setting but is typically not amenable to 
explicit solutions. In contrast, the Lagrangian approach employed in this 
article leads to tractable characterisations of the optimal coupling. A 
further contribution of this study is that unlike both 
\cite{BeiglboeckJuillet:2012} and \cite{HLTouzi:2013} who assume that $\mu$ 
(at least) is atom-free, we indicate how to deal with atoms in $\mu$. We find 
a geometric representation for the optimal coupling, which can then be applied 
to arbitrary distributions. It is likely that these ideas can be applied more 
generally.

\bibliography{biblio}

\begin{thebibliography}{10}

\bibitem{ABPS:2013}
B.~Acciaio, M.~Beiglb{\"o}ck, F.~Penkner, and W.~Schachermayer.
\newblock A model-free version of the fundamental theorem of asset pricing and
  the super-replication theorem.
\newblock {\em arXiv preprint arXiv:1301.5568}, 2013.

\bibitem{BeiglboeckHLP:2011}
M.~Beiglb{\"o}ck, P.~Henry-Labord{\`e}re, and F.~Penkner.
\newblock Model-independent bounds for option prices: A mass transport
  approach.
\newblock {\em arXiv preprint arXiv:1106.5929}, 2011.

\bibitem{BeiglboeckJuillet:2012}
M.~Beiglboeck and N.~Juillet.
\newblock On a problem of optimal transport under marginal martingale
  constraints.
\newblock {\em arXiv preprint arXiv:1208.1509}, 2012.

\bibitem{BreedenLitzenberger:78}
D.T. Breeden and R.H. Litzenberger.
\newblock Prices of state-contingent claims implicit in option prices.
\newblock {\em J. Business}, 51:621--651, 1978.

\bibitem{BrownHobsonRogers:2001a}
H.~Brown, D.G. Hobson, and L.~C.~G. Rogers.
\newblock The maximum maximum of a martingale constrained by an intermediate
  law.
\newblock {\em Probab. Theory Related Fields}, 119(4):558--578, 2001.

\bibitem{BrownHobsonRogers:2001b}
H.~Brown, D.G. Hobson, and L.~C.~G. Rogers.
\newblock Robust hedging of barrier options.
\newblock {\em Math. Finance}, 11(3):285--314, 2001.

\bibitem{CarrLee:2010}
P.~Carr and R.~Lee.
\newblock Hedging variance options on semi-martingales.
\newblock {\em Finance and Stochastics}, 14(2):179--208, 2010.

\bibitem{Cousot:2007}
L.~Cousot.
\newblock Conditions on options prices for absence of arbitrage and exact
  calibration.
\newblock {\em Journal of banking and Finance}, 31(11):3377--3397, 2007.

\bibitem{CoxWang:2013}
A.~Cox and J.~Wang.
\newblock Root's barrier: construction, optimality and applications to variance
  options.
\newblock {\em annals of Applied Probability}, 23(3):859--894, 2013.

\bibitem{d'AspremontElGhaoui:2006}
A.~d'Aspremont and L.~El~Ghaoui.
\newblock Static arbitrage bounds on basket option prices.
\newblock {\em Mathematical Programming, Series A}, 106(3):467--489, 2006.

\bibitem{DavisHobson:2007}
M.H.A. Davis and D.G. Hobson.
\newblock The range of traded option prices.
\newblock {\em Mathematical Finance}, 17(1):1--14, 2007.

\bibitem{DDGKV:2002}
J.~Dhaene, M.~Denuit, M.J. Goovaerts, R.~Kaas, and D.~Vyncke.
\newblock The concept of co-monotonicity in actuarial science and finance:
  applications.
\newblock {\em Insurance: Mathematics and Economics}, 31:133--161, 2002.

\bibitem{DolinskySoner:2013}
Y.~Dolinsky and H.M. Soner.
\newblock Martingale optimal transport and robust hedging in continuous time.
\newblock arXiv:1208.4922, 2013.

\bibitem{GalichonHLPTouzi:2013}
A.~Galichon, P.~Henry-Labordere, and N.~Touzi.
\newblock A stochastic control approach to no-arbitrage bounds with given
  marginals, with an application to lookback options.
\newblock {\em Annals of Applied Probability, to appear}, 2013.

\bibitem{HLTouzi:2013}
P.~Henry-Labordere and N.~Touzi.
\newblock An explicit martingale version of brenier's theorem.
\newblock {\em arXiv preprint arXiv:1302.4854}, 2013.

\bibitem{Hobson:1998}
D.G Hobson.
\newblock Robust hedging of the lookback option.
\newblock {\em Finance and Stochastics}, 2:329--347, 1998.

\bibitem{Hobson:2011}
D.G. Hobson.
\newblock The {S}korokhod embedding problem and model independent bounds for
  option prices.
\newblock In {\em Paris-Princeton Lecture Notes on Mathematical Finance}.
  Springer, 2010.

\bibitem{HobsonKlimmek:2011}
D.G. Hobson and M.~Klimmek.
\newblock Model-independent hedging strategies for variance swaps.
\newblock {\em Finance and Stochastics}, 16(4):611--649, 2012.

\bibitem{HobsonLaurenceWang:2005}
D.G. Hobson, P.~Laurence, and T-H. Wang.
\newblock Static-arbitrage upper bounds for the prices of basket options.
\newblock {\em \em Quantitative Finance}, 5:329--342, 2005.

\bibitem{HobsonNeuberger:2008}
D.G. Hobson and A.~Neuberger.
\newblock Robust bounds for forward-start options.
\newblock {\em Mathematical Finance}, 22(1):31--56, 2008.

\bibitem{Kahale:2009}
N.~Kahale.
\newblock Model-independent lower bound on variance swaps.
\newblock {\em Available at SSRN 1493722}, 2011.

\bibitem{Kantorovich:1975}
L.V Kantorovich.
\newblock Lecture to the memory of {A}lfred {N}obel: Mathematics in
  {E}conomics: Achievements, difficulties, perspectives.
\newblock
  www.nobelprize.org/nobel\_prizes/economics/laureates/1975/kantorovich-lectur%
e.html, 1975.

\bibitem{Ruschendorf:2007}
L.~R{\"{u}}schendorf.
\newblock Monge-{K}antorovich transportation problem and optimal couplings.
\newblock {\em Jahresbericht der DMV}, 109:113--137, 2007.

\bibitem{Strassen:67}
V.~Strassen.
\newblock Almost sure behavior of sums of independent random variables and
  martingales.
\newblock In {\em Proc. Fifth Berkeley Sympos. Math. Statist. and Probability
  (Berkeley, Calif., 1965/66)}, pages Vol. II: Contributions to Probability
  Theory, Part 1, pp. 315--343. Univ. California Press, Berkeley, Calif., 1967.

\bibitem{Villani:2009}
C.~Villani.
\newblock {\em Optimal Transport: Old and New}.
\newblock Springer-Verlag, 2009.

\end{thebibliography}

\end{document}